\documentclass[sigconf,nonacm]{acmart}

\usepackage{graphicx}
\usepackage{textcomp}
\usepackage{xcolor}
\usepackage[noend]{algpseudocode}
\usepackage{multirow}
\usepackage{setspace}
\usepackage{enumitem}
\usepackage{listings}  % for pseudo/code

\theoremstyle{definition}
\newtheorem{definition}{Definition}

\theoremstyle{definition}
\newtheorem{theorem}{Theorem}

\newcommand{\ack}{\mathsf{ack}}
\newcommand{\send}{\mathsf{send}}
\newcommand{\recv}{\mathsf{recv}}

\newcommand{\anotify}{\mathsf{anotify}}
\newcommand{\notification}{\mathsf{notification}}

\newcommand{\rcall}{\mathsf{rcall}}

\newcommand{\Pending}{\mathsf{Pending}}

\newcommand{\Completed}{\mathsf{Completed}}

\newcommand{\scope}{\nu}
\newcommand{\lock}{\mathsf{lock}}
\newcommand{\unlock}{\mathsf{unlock}}
\newcommand{\checkpoint}[1]{\overline{#1}}
\newcommand{\texec}{\mathcal{T}}
\newcommand{\True}{\mathsf{true}}
\newcommand{\False}{\mathsf{false}}

\newcommand{\elide}[1]{}

\AtBeginDocument{%
  \providecommand\BibTeX{{%
    \normalfont B\kern-0.5em{\scshape i\kern-0.25em b}\kern-0.8em\TeX}}}

\setcopyright{acmlicensed}
\copyrightyear{2024}
\acmYear{2024}
\acmDOI{XXXXXXX.XXXXXXX}

\acmConference[Conference acronym 'XX]{Make sure to enter the correct
  conference title from your rights confirmation emai}{June 03--05,
  2024}{Woodstock, NY}
\acmISBN{978-1-4503-XXXX-X/18/06}

\begin{document}

\title[Cross-Blockchain Atomicity]{Atomicity and Abstraction for Cross-Blockchain Interactions}

 \author{Huaixi Lu}
 \email{huaixil@princeton.edu}
 \authornote{Related IPs were developed when the authors were affiliated with Nokia Bell Labs.}
 \affiliation{
   \institution{Princeton University}
   \city{Princeton}
   \state{NJ}
   \country{USA}
   \postcode{08544}
 }

 \author{Akshay Jajoo}
 \email{ajajoo@cisco.com}
 \authornotemark[1]
 \affiliation{
   \institution{Cisco Research}
   \city{San Jose}
   \state{CA}
   \country{USA}
 }

 \author{Kedar S. Namjoshi}
 \email{kedar.namjoshi@nokia-bell-labs.com}
 \affiliation{
   \institution{Nokia Bell Labs}
   \city{Murray Hill}
   \state{NJ}
   \country{USA}
   \postcode{07974}
 }

\definecolor{red}{rgb}{1,0,0}

%\if 0
\newcommand{\commentaj}[1]{{\textit{\color{red} aj: #1}}}
\newcommand{\questionaj}[1]{{\textbf{{\color{red} aj QUESTION:} #1}}}
\newcommand{\todoaj}[1]{{\textit{TODO(AJ): #1}}}
\newcommand{\editaj}[2]{{{\color{red} aj EDIT?:} \sout{#1}{ \textbf{#2}}}}
\newcommand{\addaj}[1]{{\color{red}aj ADD?:} {\textbf {#1}}}
%\fi

\if 0
\newcommand{\commentaj}[1]{}
\newcommand{\questionaj}[1]{}
%\newcommand{\updated}[1]{{{}}}
\newcommand{\todoaj}[1]{}
\newcommand{\editaj}[2]{}
\newcommand{\addaj}[1]{}
\fi

\newcommand{\knote}[1]{{\color{orange}{[K]#1}}}
\newcommand{\hnote}[1]{{\color{blue}{[H]#1}}}

\begin{abstract}
  A blockchain facilitates secure and atomic transactions between mutually untrusting parties on that chain.   
Today, there are multiple blockchains with differing interfaces and security properties. 
Programming in this multi-blockchain world is hindered by the lack of general and convenient abstractions for cross-chain communication and computation. Current cross-chain communication bridges have varied and low-level interfaces, making it difficult to develop portable applications. Current methods for multi-chain atomic transactions are limited in scope to cryptocurrency swaps. 

This work addresses these issues. We first define a uniform, high-level interface for communication between chains. Building on this interface, we formulate a protocol that guarantees atomicity for general transactions whose operations may span several chains. We formulate and prove the desired correctness and security properties of these protocols. 
Our prototype implementation is built using the LayerZero cross-chain bridge. Experience with this implementation shows that the new abstractions considerably simplify the design and implementation of multi-chain transactions. Experimental evaluation with multi-chain swap transactions demonstrates performance comparable to that of custom-built implementations.

\end{abstract}

% \begin{CCSXML}
% <ccs2012>
%    <concept>
%        <concept_id>10002978.10002986.10002987</concept_id>
%        <concept_desc>Security and privacy~Trust frameworks</concept_desc>
%        <concept_significance>500</concept_significance>
%        </concept>
%    <concept>
%        <concept_id>10002978.10003006.10003013</concept_id>
%        <concept_desc>Security and privacy~Distributed systems security</concept_desc>
%        <concept_significance>500</concept_significance>
%        </concept>
%    <concept>
%        <concept_id>10003033.10003039.10003041.10003043</concept_id>
%        <concept_desc>Networks~Formal specifications</concept_desc>
%        <concept_significance>500</concept_significance>
%        </concept>
%  </ccs2012>
% \end{CCSXML}

% \ccsdesc[500]{Security and privacy~Trust frameworks}
% \ccsdesc[500]{Security and privacy~Distributed systems security}
% \ccsdesc[500]{Networks~Formal specifications}

% \keywords{cross-chain interoperability, smart contracts, formal analysis, atomicity, protocol}

\maketitle
\section{Introduction} \label{Introduction}

A blockchain provides the abstraction of an immutable, open, and fault-tolerant ledger. Several blockchains with a variety of features have been designed and implemented since the introduction of the Bitcoin chain~\cite{nakamoto2008bitcoin}. The rich variety comes at a cost: information is  spread across chains, complicating the development of applications. 

From the programmer's viewpoint, there is a lack of effective abstractions for inter-chain communication and computation. Current "cross-chain bridge'' mechanisms facilitate basic communication between two chains but expose varied and low-level interfaces, making it challenging to write easily portable applications. Smart contract transactions execute atomically on a single chain, i.e., with ``all or nothing'' semantics. However, distinct chains operate concurrently and independently of one another; hence, atomicity does not extend across  chains and must be enforced by a protocol. Currently, atomic cross-chain protocols are known only for cryptocurrency swaps: a useful application but with a narrow scope. 

There is thus a need for convenient, broadly applicable programming abstractions for computation and communication across blockchains. In this work, we design two such mechanisms: (1) a uniform interface for pairwise blockchain communication, and (2) a protocol that ensures atomicity for transactions whose operations may span several chains. 

The interface provides three commonly-needed communication functionalities: notification (sending a message from one chain to another); notification with acknowledgement of receipt; and a remote procedure call. The central advantage of this interface is portability: smart contracts written to this interface execute without modification on different chains connected by different bridges, much as the standard \texttt{socket} library~\cite{gilligan1999basic} of Unix allows programs to be ported to a variety of network types and interfaces. We develop an architecture for realizing these functionalities and instantiate it over the LayerZero and IBC cross-chain bridges. The  interface also simplifies programming, eliminating bugs such as one that we discovered in LayerZero's own example PingPong application. 

As information spreads out over multiple chains, there is a need for transactions whose operations span several chains. The resulting technical challenge, which we address in this work, is to provide a mechanism that guarantees atomic and trustworthy execution. %On a single chain, transactions are atomic thanks to the consensus mechanism, which ensures that a transaction is committed only if a majority \knote{majority is correct?} \anote{I agree, I don't think we can and should say "Majority". I am making an attempt at rewriting it.} of the participants agree on its validity. 
%%On a single chain, its consensus mechanism guarantees the atomicity of transactions. 
%are atomicthis is guaranteed by the underlying well defined consensus mechanism. The consensus mechanism ensures that a transaction meets a set of pre-defined requirements, } \hnote{such as a certain number of agreements within participants, } 
%\addaj{for it to come into effect. Once a transaction comes into effect all nodes accept it religiously.}
%%However, multiple chains  operate concurrently and independently of each other and therefore atomicity must be imposed through a protocol. This task is  complicated by differences in block formats and consensus algorithms. 
%are not synchronized: their executions are independent and concurrent
%work independently and are concurrently active \addaj{and they have their own consensus algorithms}. Despite the trust-less decentralized setting, the trustworthiness of a transaction on a single chain is ensured (with high probability) through replication and consensus. That is, once a transaction is committed on the blockchain, it is guaranteed with high probability that it was correctly executed. However, there is no such guarantee across chains.

We design a protocol for atomic execution of general cross-chain transactions. This protocol is inspired by the classic two-phase commit protocols in distributed databases, and operates over the uniform communication interface. We prove atomicity and trustworthiness; I.e., we show that the protocol cannot be disrupted by malicious behavior by other parties. 
Trustworthiness relies on a secure-transfer property of cross-chain bridges, as well as the presence of trusted contracts that help execute remote operations. 

We have developed a prototype implementation (in Solidity) using the LayerZero and IBC bridges. We have deployed this implementation on testnets linked by LayerZero bridges and used the atomicity protocol to implement and successfully test pairwise and multi-way cryptocurrency swaps. Multi-chain cryptocurrency swaps are easily expressed in the general model, requiring only $10-20$ lines of code to define a transaction.

To summarize, this work makes the following  contributions:
\begin{enumerate}
    \item We define a uniform interface for inter-blockchain communication, ensuring portability of applications across a variety of cross-chain bridges. 
    
    \item We define a protocol for atomic execution of general, cross-chain transactions. 
    
    \item We show the correctness and trustworthiness of these protocols, under clearly stated assumptions on blockchain and bridge computation.
    
    \item Experience with programming multi-chain swap transactions shows that the abstractions considerably simplify the design of transactions, while the prototype implementation has performance that is comparable to that of custom-built implementations. 
\end{enumerate}

This paper is organized as follows. We begin with relevant background on blockchain and blockchain interoperability. This is followed by the design of the uniform communication interface  and the multi-chain atomicity protocol, supported by formal modeling and proofs of correctness and trustworthiness. We then describe an evaluation  on case studies that utilize  LayerZero as the low-level bridge communication between chains. Finally, we discuss related work and draw conclusions along with suggestions for future work.

\section{Background}

\subsection{Blockchain Interoperability}
A blockchain system is a distributed ledger that maintains a continuously growing history of unalterable ordered information. The information is organized in a  chain of discrete \emph{blocks}. The potential benefits of the blockchain are more than just economic-they extend into political, humanitarian and scientific domains~\cite{swan2015blockchain}. %chronological [kedar?]

Smart contracts are digital agreements coded and stored within the blockchain~\cite{zheng2020smartcontract}. Unlike 
conventional paper-based contracts, smart contracts encode terms and conditions through software. In this inherently decentralized, trustless setting, the immutability of contract code and openness of contract execution (both ensured by the underlying blockchain) develop trust between the parties to a contract. The blockchain mechanism ensures that every smart contract transaction is atomic: i.e., it either completes successfully or (on failure) leaves the blockchain state unchanged. 

% ensuring 
% their execution aligns with the consensus mechanism inherent to the blockchain. One action defined in the smart contract refers to one transaction that can be recorded as a piece of information in a single block.

A wide range of blockchains with diverse use-cases and diverse consensus protocols have been deployed 
%\addaj
{and are in active use}. 
For example, Ethereum supports smart contracts and uses Proof of Stake (PoS)~\cite{nguyen2019pos} while Bitcoin is built as a cryptocurrency on Proof of Work (PoW)~\cite{gervais2016pow}. 
%\editaj{This diversity has precipitated a significant surge in the heterogeneity of the blockchain ecosystem.}
%{When catering such diverse use cases the underlying ecosystem gets heterogeneous.}
Looking forward, %\editaj{it is highly probable that we will see the emergence of}
{this heterogeneity will only increase and will lead to} a 'multi-chain universe'—a digital environment where data is spread across blockchains and applications combine operations that execute on multiple chains. 
%wherein 
%\editaj{disparate blockchain applications can seamlessly communicate, which is known as blockchain interoperability.}
%{applications across disparate blockchains will seamlessly communicate.}

\subsection{Cross-Chain Bridges} \label{sec_ccc_bridge}
To facilitate easy, secure, and accurate interactions between distinct blockchains, several Cross-Chain Communication (CCC) bridges have been established. A fundamental feature of these bridges is the \textit{Secure Transfer} property, which is described in many works such as~\cite{xie2022zkbridge} and formalized later in §III. This property requires that when a receiver chain receives a notification of an event on the sender chain, this event has actually been 
%authenticate that a specific event has indeed been
%{ensure that the event under consideration has been}
committed on the sender chain and cannot be rolled back.
Without the \textit{Secure Transfer} property, the sender chain may create a fake message that is not committed on the sender chain. As a result, a user on the receiver chain would 
lose some assets without a proper counter transaction on the sender chain, thereby creating a fraudulent transaction.
%reaping an unjustified benefit.

Figure~\ref{fig:cc_bridge} demonstrates the general process for transmitting messages via a CCC bridge.
\begin{figure}
    \centering
    \includegraphics[width=0.88\linewidth]{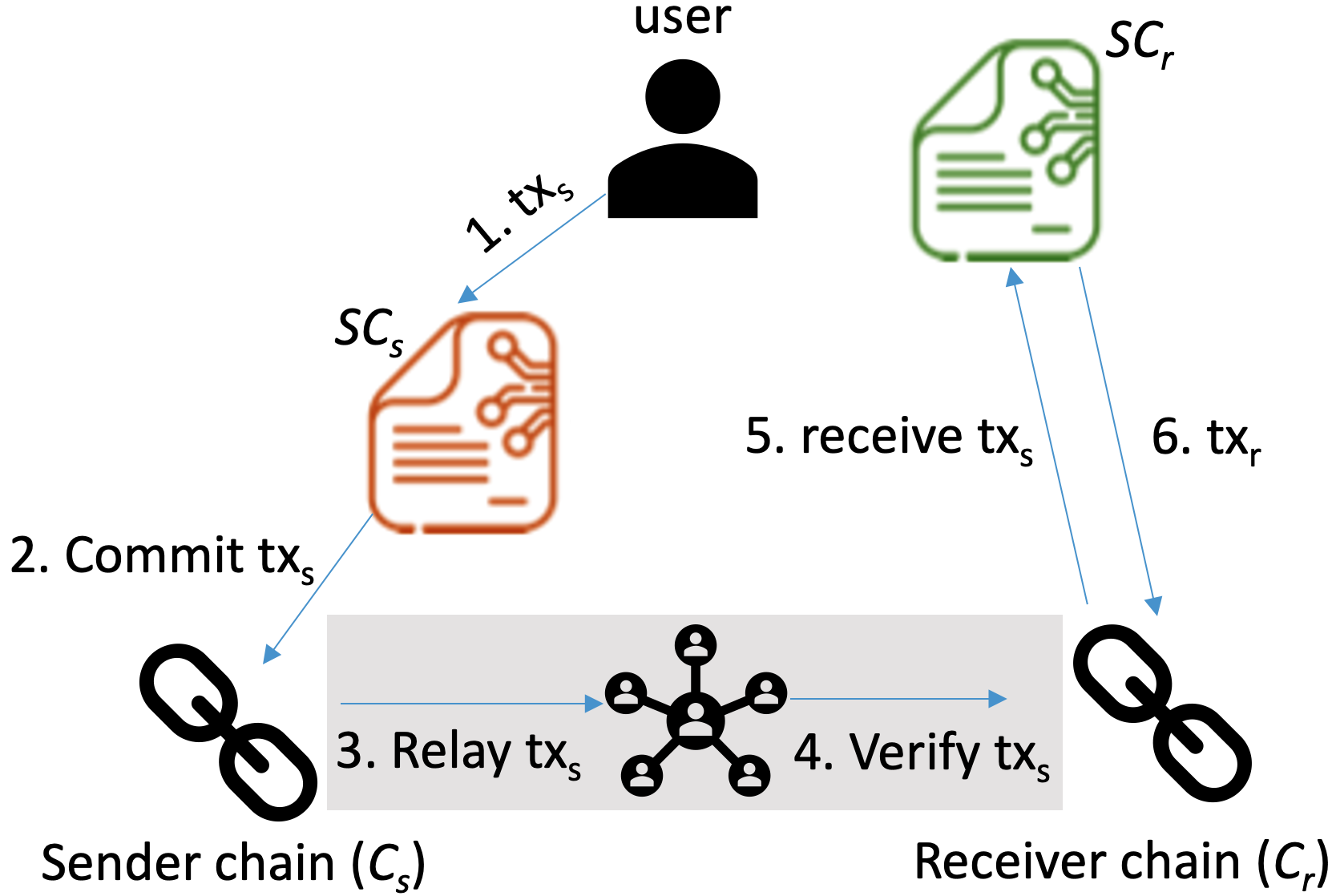}
      \caption{\label{fig:cc_bridge} \textbf{The communication flow through a CCC bridge}}

\end{figure}
Smart contracts $SC_s$ and $SC_r$ are deployed on blockchains $C_s$ (sender) and $C_r$ (receiver), respectively. Note that a single bridge is usually designed for one direction. For a user to activate a transaction $tx_r$ on blockchain $C_r$, they must first initiate transaction $tx_s$ on blockchain $C_s$.
(The remote execution may be motivated by superior features on $C_r$, such as lower fees, prompting the user to transfer assets from $C_s$ to $C _r$.)
The secure transfer property is guaranteed by the middle CCC bridge (symbolized by the shaded area in 
Figure~\ref{fig:cc_bridge}). The bridge relays the details of $tx_s$ and confirms its existence on chain $C_s$ before transmitting it to chain $C_r$. 
The above communication flow can be summarized as: 
1) User initiates $tx_s$ in smart contract $SC_s$. 
2) $tx_s$ undergoes the consensus protocol of chain $C_s$ and is subsequently committed. 
3) The CCC bridge picks up the message including $tx_s$ and sends it to the receiver chain $C_r$, facilitated by one or multiple relayers. 
4) The CCC bridge validates $tx_s$ by examining the associated verified block-header. 
5) $tx_s$ is confirmed, and its details are available on Chain $C_r$. 
6) Consequently, smart contract $SC_r$ issues $tx_r$ accordingly which is then committed on $C_r$.

%CCC bridges can be categorised into two distinct types, namely ``trusted" and ``trustless", based on their varying trust assumptions and security models to verify the transaction from the sender chain.
Based on their trust assumptions and security models to verify the transaction from the sender chain we can categorize the CCC bridges in the following two categories:
\begin{itemize}
    \item \textbf{Trusted:} The ``trusted" CCC bridges~\cite{zarick2021layerzero}, introduce a third-party trust system, such as an Oracle network or a notary committee, to authenticate and verify the occurrence of specific events on an alternative chain.
    \item \textbf{Trustless:} The ``trustless" CCC bridges, e.g. Cosmos inter-blockchain communication protocol (IBC)~\cite{goes2020interblockchain} and zkbridge~\cite{xie2022zkbridge}, on the other hand, employ a `light client' approach, where one chain maintains an up-to-date version of the other chain's block header. For instance, a light client of the sender chain would be implemented within the receiver chain. This method allows the receiver chain to acquire and update block header information from the sender chain, thereby verifying the occurrence of a specific event on the sender chain through a cryptographic (e.g., Merkle proof) method.
\end{itemize}
However, it should be noted that these CCC bridges primarily focus on the trust layer of the general communication between two blockchains, aiming to maintain the secure transfer property. In contrast, certain application-specific blockchain interoperability protocols, such as the cross-chain atomic swap enabled by Hash Time-locked Contract (HTLC)~\cite{herlihy2018atomic}, concentrate solely on a distinct communication domain between the two chains.

Given the vast variety of CCC bridges, there have been efforts in prior research to combine or standardize all these bridges such as Multi-Message Aggregation (MMA)~\cite{2023mma}. While these approaches tend to blur the boundaries among different bridges, they work at the same low-level as cross-chain bridges. Their main goal is to augment security based on existing bridges, rather than provide a higher-level abstraction of these low-level bridges to enhance better smart contract programmability and portability. 
The increasing complexity of emerging multi-chain applications, typically constructed on top of these low-level bridges, underscores the challenges in design and verification for correctness. A more abstract representation for these low-level bridges is indeed essential.

\subsection{Atomicity for Multi-chain Applications}
%\knote{Reword to bring out how atomicity or multi-chain operations are done today and point out drawbacks}
%In a single blockchain, atomicity is achieved via the consensus protocol, which ensures that one transaction consisting of a set of operations is either fully committed to a new block or entirely rejected from being appended to the chain

The fundamental property of the consensus protocols in a blockchain is that a \emph{block} is either committed or discarded. This property can be leveraged to achieve atomicity across a set of operations. If we embed all the operations under consideration in a same block then all of them will be committed or all of them will be rejected; \textit{i.e.} the set of operations will be executed atomically.

Informally speaking, a cross-chain transaction is a sequence of operations where the operations may target different chains. For instance, a cross-chain transaction $T$ may be expressed as the sequence $[c_0;d_0;c_1]$, where $c_0,c_1$ are operations on a chain $C$ and $d_0$ is an operation on a different chain $D$. Atomic execution of such a transaction is not guaranteed. That is because each blockchain operates independently of the others, without a common consensus mechanism. Thus, the sequence of operations in transaction $T$ has no built-in protection against interference from other operations on those chains, which may be initiated independently at any time by other parties. (For example, an operation $c_2$ may occur on chain $C$ between the occurrences of $c_0$ and $c_1$, changing the state of blockchain $C$ in a way that cannot be viewed as the result of an atomic execution -- i.e., as the result of $(c_0;d_0;c_1);c_2$ or $c_2;(c_0;d_0;c_1)$.) It is for this fundamental reason that achieving cross-chain atomicity requires a special protocol.

Existing protocols for atomic cross-chain transactions are limited in scope: e.g., pairwise token transfers through gateways~\cite{augusto2023multi} or pairwise asset exchanges through atomic swap protocols~\cite{ding2022lilac, zakhary13atomic}. HTLC~\cite{herlihy2018atomic} is the most popular solution for ensuring atomicity of a pairwise cryptocurrency swap. It has been extended to multiple chains~\cite{xue2023fault}, but remains specialized to currency swaps, an important but narrowly defined transaction. As multi-chain applications evolve to handle data spread across several chains, there is a need for an atomicity protocol that handles general transactions that contain arbitrary smart-contract operations defined on multiple chains. That is the fundamental question addressed in this paper. 

%%However, when it comes to cross-chain communication, atomicity is not always guaranteed. This is because each blockchain operates independently, without a unified consensus across distinct chains. Most interoperability solutions to date have concentrated on cross-chain communication involving just two parties or are limited to token transfers through gateways~\cite{augusto2023multi} and asset exchanges through atomic swap protocols~\cite{ding2022lilac, zakhary13atomic}. As multi-chain applications evolve, the need for an atomicity mechanism for general cross-chain transactions 
%reaching a consensus across all blockchains for atomicity 
%challenge.

%A top level consensus for blockchains are usually eveloped for multi-chain communication. 
Some systems use a top-level chain to periodically checkpoint the state of other chains, or to record transactions occurring on other chains in a conveniently accessible location. This is the case, for instance, for Trustboost~\cite{wang2022trustboost}, Hyperservice~\cite{liu2019hyperservice}, and Polkadot~\cite{wood2016polkadot}. This does not, however, suffice to guard against interference, and thus does not provide any guarantee of atomicity. 

%HyperService proposes a programming language for developing the decentralized applications. It aims to provide better programmability, which allows the developer to write all transactions for different blockchains in a single script. It generates a blockchain to connect different blockchains, but in this case, all the cross-chain communications are built on the trust for the single blockchain.

\section{Threat Model}\label{sec:threats}

We describe and discuss our threat model. We first present our correctness and security goals. This is followed by listing the attacker’s capabilities which may compromise those goals. Finally, we formulate the assumptions on which the trustworthiness claims are based. 

The correctness objective of the communication interface is defined by the semantics of each high-level operation. Consider the notification operation. The desired safety (security) property is that a notification message received by a contract  $d$ on a chain $D$ that is claimed to have been sent by a contract $c$ on a chain $C$ must have actually been sent by contract $c$: i.e., notifications cannot be forged. There is also a liveness requirement of eventual delivery: if contract $c$ on chain $C$ sends a notification to contract $d$ on chain $D$, that notification must eventually be delivered to $d$. Similar safety and liveness properties apply to the other operations: notification with acknowledgement and remote call.

The correctness objective of the atomicity protocol can also be divided into safety and liveness objectives. The safety objective is that of atomicity, which can be expressed as strict serializability: roughly, that on every computation, the individual atomic operations of transactions may be reordered so that the transactions do not overlap (i.e., appear to execute atomically), with transactions that do not originally overlap retaining their temporal ordering (i.e., if, in the original computation, a transaction $T_1$ completes before a transaction $T_2$ is started, this temporal relation must be preserved after reordering). In terms of the closely related concept of linearizability, every transaction should appear to occur instantaneously at a singular point in time, falling between its start and end points. The liveness objective is that the execution of every transaction should eventually terminate (either successfully or with failure). There is a secondary safety objective, which is to limit the number of failing transactions.

An attacker's goal is to disrupt these correctness properties: that is, either to interfere with transaction execution so as to violate atomicity (i.e., make it impossible to reorder transactions correctly), or to interfere in such a way that an initiated transaction never completes, or to interfere by causing transactions to fail (i.e., create a denial of service attack). 
%%In this work, we identify the attacker's primary goal as disrupting the ongoing cross-chain transaction within a reasonable amount of time. This disruption aims to alter the outcome of the correctly executed cross-chain transaction (specified in Figure~\ref{fig:tx_spec}).

The capabilities attributed to an attacker include:
\begin{itemize}
    \item The ability to invoke any methods in a smart contract, as long as the attacker has the required permissions.
    \item The freedom to execute these methods at any point during a cross-chain transaction. 
    \item The flexibility to assume any role in the transaction process, be it a participant on a blockchain, or a relayer in a cross-chain bridge. 
\end{itemize}
However, it is assumed that the attacker {\bf cannot} steal participants' private keys or arbitrarily alter the state of a blockchain ledger, e.g., by compromising the underlying consensus mechanism.

We assume that the operation of a blockchain and its smart contracts cannot be compromised. That is, a smart contract executes precisely according to its definition and cannot halt due to machine failure. This assumption is necessary as our protocols rely on  auxiliary smart contracts. The protocols also rely on cross-chain bridges, which we assume  meet the secure transfer property that is discussed in Section~\ref{sec_ccc_bridge}. The secure transfer property must hold regardless of the particulars of the cross-chain bridge protocol, which differs significantly betweeen bridge implementations such as zk-bridge~\cite{xie2022zkbridge}, IBC by Cosmos~\cite{goes2020interblockchain} and LayerZero~\cite{zarick2021layerzero}.

\section{An Abstract Communication Interface}

We present a high-level interface that a pair of blockchains can use to communicate over a cross-chain bridge. As sketched in the \S\ref{Introduction}, the primary objective of this interface is portability: application software written to this interface can be easily ported from one cross-chain bridge to another and 
%(subject to native language compatibility \commentaj{will writing this point here weaken our case?}) 
(with native bytecode compatibility) from one blockchain to another. That is not the case for applications that are written to the interface of a specific cross-chain bridge, as the application is then locked to this interface, and the interfaces are incompatible with each other. We show how our high-level interface can be implemented atop the lower-level mechanisms provided by current bridges such as LayerZero.

The interface consists of three common communication functions: notify, notify with acknowledgment, and remote call. We describe each in turn. Consider blockchains $C$ and $D$ 
%%\commentaj{(Why writing it as chain $C$ and $D$? Why not $A$ and $B$? After reading $C$ and $D$. I immediately start searching for $A$ and $B$ also.)} 
% [kedar] C for chain, D for different chain? Good point, but too much work to modify it now.
that are linked bidirectionally by cross-chain bridges. The notify function is used by a contract $c$ on chain $C$ to send a message $x$ to a contract $d$ on chain $D$. The notify-with-acknowledgement functionality is used by contract $c$ to send a message $x$ to contract $d$ \emph{and} receive an acknowledgement that this message has been delivered successfully. The remote-call functionality is used by contract $c$ to invoke a method $m$ of contract $d$ with parameters $x$ and receive an acknowledgement that this invocation was performed successfully along with values produced as a result of the invocation.

These functionalities are implemented in an \emph{asynchronous} manner, which allows contract $c$ to continue execution without waiting for the completion of the communication process. As with other asynchronous interfaces, we provide a mechanism for contract $c$ to test whether the communication is complete and to retrieve any results. It is of course easy to implement synchronous forms of communication given these capabilities. Asynchrony enables $c$ to issue remote operations on multiple chains in parallel, a capability that is used in the atomicity protocol, described in Section~\ref{Atomicity}. 

% source chain, destination chain
% sending/invoking contract, receiving contract
% action, invocation, message

We have so far used the phrase ``cross-chain bridge'' informally. Formally, a \emph{cross-chain bridge} (``bridge'' for short) from chain $C$ to chain $D$ is a communication service that is external to either chain. We abstract from specific implementations and simply assert that a contract $c$ on chain $C$ can request the bridge to transfer a message $m$ to a contract $d$ on chain $D$ through an invocation $\send(m,d)$. This invocation is recorded on chain $C$. At the other end, a bridge delivers message $m$ to contract $d$ by invoking a $\recv(m,k)$ action on $d$, where $k$ is a non-negative number.

%%. Internally, a bridge transfers a message $m$ to contract $d$ as a tuple $(m,d,k)$ where $k$ is an non-negative number $k$; we refer to this tuple as an indexed message. 

Given the trustless setting of Web3, the bridge service cannot be implicitly trusted to deliver messages and to do so without any corruption. Hence, we must make trust assumptions on the behavior of the bridge. We refer to those assumptions as the \emph{Secure Transfer} property of the bridge. The formalized formulation is adapted from the system model presented in~\cite{xie2022zkbridge}. In practice, the assumptions are enforced (with high probability) through cryptographic techniques such as zero-knowledge proofs.

\begin{definition}[Secure Transfer](adapted from the model in ~\cite{xie2022zkbridge})
  A bridge connecting chain $C$ to chain $D$ satisfies the secure transfer property if it meets the following conditions:
  \begin{description}[leftmargin=*]
  \item[{\bf Safety:}]   Every message $m$ that is delivered to a contract $d$ on chain $D$ was indeed sent from chain $C$ -- i.e., messages cannot be forged by the bridge. This may be formally defined as:
    \[ \forall m,k,d,i : d.\mathrm{recv}(m,k) \in D_i \Rightarrow c.\mathrm{send}(m,d) \in C_k
    \]
    Precisely, for every message $m$ and number $k$ such that there is a $\recv(m,k)$ invocation by contract $d$ on $D_i$, the $i$'th block of $D$: (1) $k$ is the index of a stable block on chain $C$ and (2) $\send(m,d)$ is an invocation from contract $c$ on chain $C$ to the bridge that is recorded on the $k$'th block of $C$, denoted as $C_k$.
    
  \item[{\bf Liveness:}] Every message $m$ sent to a valid contract $d$ on chain $D$ is eventually delivered to $d$, which can be formally described as: 
  \[ \forall m,k,d : c.\mathrm{send}(m,d) \in C_k \Rightarrow \exists i : d.\mathrm{recv}(m,k) \in D_i \]
    
  Precisely, for every message $m$ and every invocation $\send(m,d)$ from a contract $c$ on chain $C$ to the bridge, there is eventually an invocation $\recv(m,k)$ on contract $d$ in chain $D$ (recorded in $D_i$, the $i$'th block of $D$), where $k$ is the index of the block on chain $C$ (denoted as $C_k$) that records the $\send(m,d)$ invocation.

  Note that messages could be delivered out of order. I.e., it is possible for the send of $m_1$ to be recorded on block $k_1$ and that for $m_2$ to recorded on block $k_2$ where $k_1 < k_2$, while the corresponding deliveries are made at blocks $i_1,i_2$ where $i_1 \geq i_2$. 

  \end{description}

\end{definition}

%% \begin{enumerate}
%% \item Notify. An asynchronous notification message conveys data $x$ from a sending contract $c$ on $C$ to a receiving contract $d$ on bridge $D$. 

%% \item Notify with acknowledgement. This is an asynchronous notification action that returns a ``future'' containing an acknowledgement for the notification receipt. This future can be blocked upon or inspected by the sending contract. 

%% \item Remote call. This is an asynchronous invocation of a method in a contract on chain $D$, which returns a future containing the results of the method invocation.
%% \end{enumerate}

%% The remote call and notification with acknowledgement are made asynchonous for increased flexibility in programming. Asynchrony allows multiple notifications to be in progress at a time and it also allows multiple remote calls (possibly to several distinct chains) to be in progress concurrently. We now present the detailed protocol for each interface action.

We now define the protocols used to implement these high level functions. The protocols require  \emph{adapter} smart contracts: contract $p$ (on chain $C$) and contract $q$ (on chain $D$). These adapters act as intermediaries between the source contract $c$ and the bridge, as well as between the bridge and the destination contract $d$, playing three critical roles. First, they provide the desired abstraction layer that insulates source and destination contracts from the low-level details of the bridge interface. Second, the adaptors ensure that notifications are delivered to and remote methods are invoked on the destination contract. And third, the adaptors track the progress of the asynchronous communications operations by, for instance, associating internal fresh identifiers (e.g., sequence numbers) with operations and responding to queries from the source contract.

To limit duplication, we present the protocols for notify-with-acknowledgement and remote-call. The notify protocol is a sub-protocol, constituted by the initial steps of the notify with acknowledgment protocol. For simplicity of notation, we drop the block index from the $\recv(m,k)$ invocation, representing it simply as $\recv(m)$.

\subsection{Notify with acknowledgement}

\begin{figure*}
 \vspace{-30pt}
  \centering
  \includegraphics[scale=0.55]{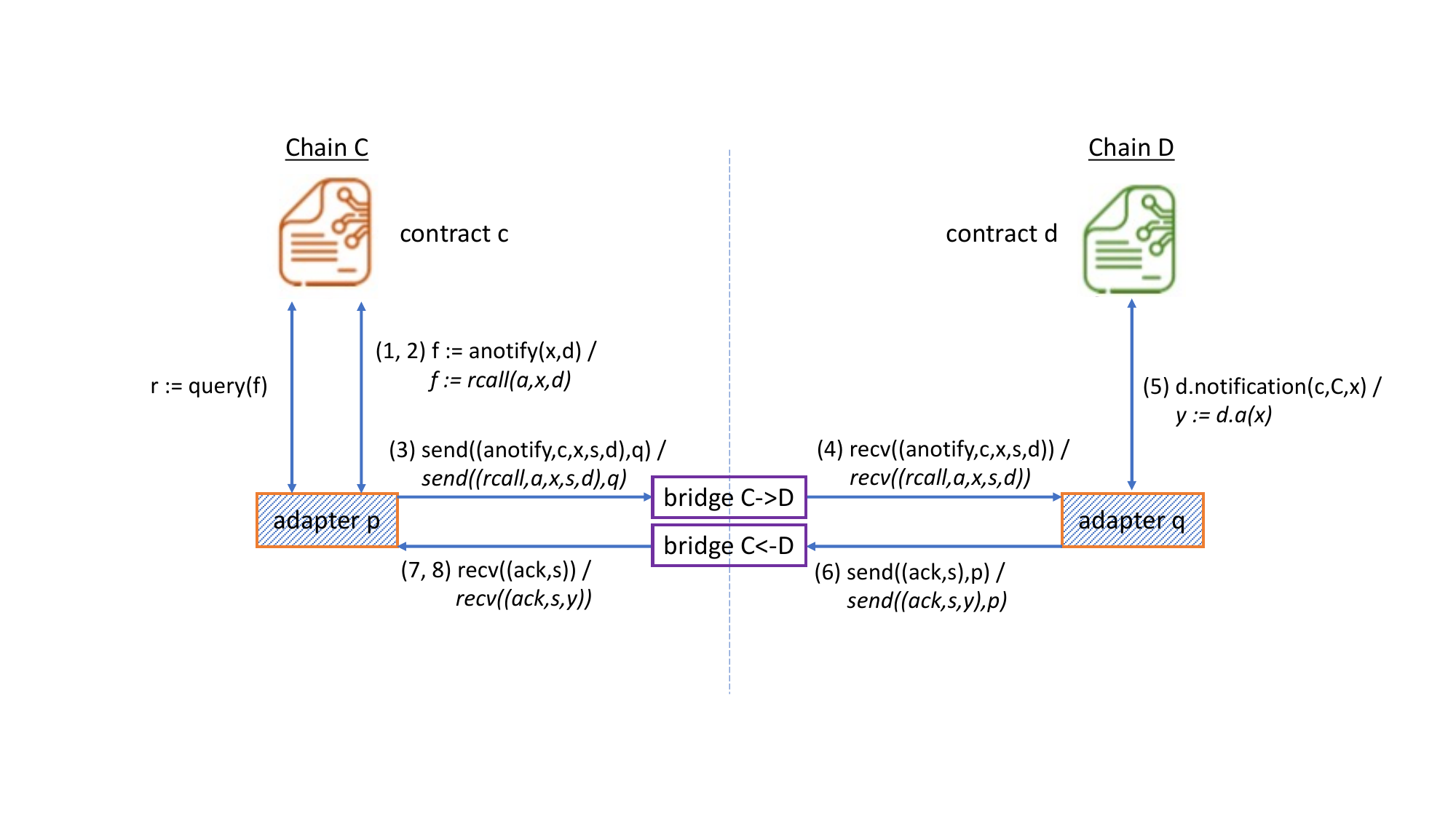}
  \vspace{-60pt}
  \caption{\bf A Combined Flow Diagram for the Notify-with-Acknowledgement/Remote Call Protocol.
%  \ajcomment{Change the "bridge D->C" to "bridge C<-D"}
}
  \label{fig:notify-with-ack-remote-call}
\end{figure*}

This is a notification that also provides an acknowledgement of receipt. In order to generate the acknowledgement, we need a bridge in the reverse direction, from $D$ to $C$. Acknowledgements are provided asynchronously; thus, it is necessary to mark messages with sequence numbers to distinguish those that have been acknowledged. The protocol is illustrated in Figure~\ref{fig:notify-with-ack-remote-call} along with the Remote Call Protocol and consists of the following sequence of steps. 

\begin{enumerate}
\item Contract $c$ invokes $\anotify(x,d)$ on the adapter contract $p$ on chain $C$, where $x$ is the data to be notified to contract $d$ on chain $D$.
\item Adapter $p$ assigns a fresh identifier $s$ and a  ``future'' $f$ (with state \textsf{Pending}) to this invocation, and returns $f$ to contract $c$. Contract $c$ can check the status of the transmission at any later point by querying the future $f$ through $r:= query(f)$. (The adapter can also provide a facility for contract $c$ to be notified when the future changes state from \textsf{Pending} to \textsf{Delivered}. We do not illustrate this in the figure.)
  
\item Adapter $p$ invokes $\send(m=(\anotify,c,x,s,d),q)$ on the bridge from $C$ to $D$, where $q$ is its corresponding adapter contract on chain $D$. 
\item The bridge communicates the message $m$ to contract $q$ through its $\recv$ method. 
\item Adapter $q$ parses $m$ as $(\anotify,c,x,s,d)$ and notifies contract $d$ by invoking its $\notification$ method with parameters $c$, $C$ and $x$.
\item Adapter $q$ then prepares a message $m'=(\ack,s)$ and sends it to the bridge for the reverse direction from $D$ to $C$ through $\send(m',p)$.
\item The bridge communicates the message $m'$ to contract $p$ through its $\recv$ method.
\item Adapter $p$ parses $m'$ as $(\ack,s)$ and updates the state of the future $f$ associated with $s$ to \textsf{Delivered}.
\end{enumerate}

At any point, contract $c$ may query the future $f$ for the state of the notification. That will be \textsf{Pending} until the final step, when the state changes to \textsf{Delivered} and is then unchanged.

The notify protocol is a simplification and shortening of this protocol. It consists of steps 1-5 of this protocol, i.e., without the additional acknowledgement message transfer. There is also no need to have a future $f$ and fresh identifier $s$ as those are used exclusively for the acknowledgement step. 

\begin{theorem} \label{thm:notify-ack}
    The protocol ensures that notifications cannot be forged (safety) and are eventually delivered (liveness). 
\end{theorem}
\begin{proof} 
The proof relies on the assumptions that the working of a blockchain and the behavior of a contract cannot be tampered with, and on the secure transfer property of the bridge. 

  (Liveness) Consider data $x$ that is sent from contract $c$ to contract $d$. By the working of chain $C$, the adapter contract $p$ will eventually execute the $\anotify$ call with the message $m$ as defined. From the liveness property of Secure Transfer, message $m$ will be eventually delivered to contract $q$. By the working of chain $D$, the $\notification$ call will be executed on contract $d$, delivering the notification for the data $x$. 
  
  (Safety) Now consider a notification delivered to $d$ from the adapter $q$. From the definition of the adapters $p$ and $q$ and the soundness of Secure Transfer,  the notification message must originate from a valid contract on chain $C$ and has not been tampered with in transit. 
  
  A similar argument in reverse establishes the security and eventual delivery of the acknowledgement message.
\end{proof}

This is an asynchronous protocol, in that multiple unacknowledged notifications may be issued by the contract $c$. The identifier $s$ is used to track the state of each notification, which is reflected in the future $f$. The smart contract program $c$ may inspect and make decisions based on whether an issued notification has been acknowledged; for instance, blocking any other transactions until a specific notification is acknowledged. As a special case, it is easy to implement synchronous versions of notify and notify-with-acknowledgement.

\subsection{Remote Call}

%\begin{figure*}
%  \centering
%  \includegraphics[scale=0.55]{figures/remote-call-flow-diagram.pdf}
%  \caption{Flow Diagram for the Remote Call Protocol.}
%  \label{fig:remote-call}
%\end{figure*}

A remote call invokes a particular method on the target contract. The protocol extends the one for notify-with-acknowledgement by carrying additional information in the messages that identifies the method and its return value. The protocol is illustrated in Figure~\ref{fig:notify-with-ack-remote-call} and consists of the following sequence of steps. 

\begin{enumerate}
\item Contract $c$ invokes $\rcall(a,x,d)$ on the adapter contract $p$ on chain $C$, where $a$ identifies the method on contract $d$ on chain $D$ that is to be executed and $x$ represents the parameters to $a$.
  
\item Adapter $p$ assigns a fresh identifier $s$ and a  ``future'' $f$ (with state $\Pending$) to this invocation, and returns $f$ to contract $c$. Contract $c$ can check the status of the remote call  at any later point by querying the future $f$ through $r:= query(f)$. (The adapter can also provide a facility for contract $c$ to be notified when the future changes state from \textsf{Pending} to \textsf{Completed}. We do not illustrate this in the figure.)
  
\item Adapter $p$ invokes $\send(m=(\rcall,a,x,s,d),q)$ on the bridge from $C$ to $D$, where $q$ is its corresponding adapter contract on chain $D$. 
\item The bridge communicates the message $m$ to contract $q$ through its $\recv$ method. 
\item Adapter $q$ parses $m$ as $(\rcall,a,x,s,d)$ and invokes method $a$ on contract $d$ with parameters $x$. Let $y$ represent the result of this invocation. 
\item Adapter $q$ then prepares a message $m'=(\ack,s,y)$ and sends it to the bridge for the reverse direction from $D$ to $C$ through $\send(m',p)$.
\item The bridge communicates the message $m'$ to contract $p$ through its $\recv$ method.
\item Adapter $p$ parses $m'$ as $(\ack,s,y)$ and updates the state of the future $f$ associated with $s$ to $\Completed(y)$. 
\end{enumerate}

At any point before the final step, contract $c$ may query the future $f$ for the state of the remote call. That will be $\Pending$ until the final step is complete, when the state changes to $\Completed(y)$ and is then unchanged.

This remote call mechanism is also asynchronous, allowing multiple remote calls to be ``in flight'' concurrently, possibly to contracts on different chains. As before, the fresh identifier $s$ is used to track the status of the remote call, which is reflected in the future $f$. The future/query mechanism can be used to turn the remote call into a synchronous one, by blocking any further actions of contract $c$ until the query returns the $\Completed(y)$ result.

Following the same line of reasoning as the proof of Theorem~\ref{thm:notify-ack}, we obtain this theorem. 
\begin{theorem} 
The protocol ensures that remote calls cannot be forged (safety) and are eventually executed according to their semantics (liveness). 
\end{theorem}

\section{An Atomic Multi-Chain Transaction Protocol}\label{Atomicity}

On a blockchain, any single invocation of a smart contract is atomic by design. However, there is no guarantee that a sequence of invocations executes as a unit; i.e., without interference from actions on that blockchain that may be initiated independently by other parties. There is a similar lack of atomicity when considering a sequence of invocations to  contracts that belong to different blockchains. Each blockchain executes independently of the others, so there is no protection from interference by actions that may be independently initiated by other parties.

But there are interesting use cases that require atomicity: for instance, two-way or multi-way cryptocurrency exchanges between tokens on different blockchains. The exchange must be executed as a unit, without interference from actions that could, for instance, deplete the account balances of the parties involved. In this section we present our atomicity protocol, which ensures that a cross-blockchain transaction either executes successfully or, if a constituent operation fails, then all operations are canceled and the effect is as if there was no change to the state of the participating blockchains. 

We begin by defining a blockchain as an abstract state machine; formulate a cross-chain transaction; describe the atomicity protocol; then establish its correctness.

\subsection{Cross-Chain Transactions}

For the purposes of defining the atomicity protocol, we may view a blockchain $B$ as an abstract state machine $(V,\Sigma,T)$, where: 
\begin{itemize}
\item $V$ is a set of typed state variables. This induces a state space $S$ which is the set of assignments of values to variables in $V$. (Notationally, we can say $S = V \rightarrow D$, where $D$ is the type domain. In the formal presentation, all variables have a single domain for simplicity.) The connection to a real blockchain is that $V$ represents the state variables for all contracts on the blockchain. Over time, smart contracts are added to the chain, thus extending $V$ with fresh state variables, but we can ignore this aspect for the formal presentation.

\item $\Sigma$ is a set of action names. In a real blockchain, each action represents a contract method along with values for its parameters. 

\item $T: S \times \Sigma \to S \times D$ is the deterministic transition function of the chain, which transforms the current state, say $s$, to a successor state $t$ on action $a$ ($ a \in$ $\Sigma$) and produces a result value. In an actual blockchain, this represents the execution of a smart contract method, which must be deterministic.

\item We associate a \emph{scope} with each action through a function $\scope: \Sigma \to 2^{V}$. The scope $\scope(a)$ of action $a$ is the subset of variables that is affected (read or written to) on a transition labeled by $a$.

\item An \emph{indexed} action is a pair $(B,a)$ where $B$ is a blockchain and $a$ is an action in $B$. 
\end{itemize}

A \emph{cross-chain transaction} is specified by a finite set of indexed actions that are related by an irreflexive partial order denoted $\prec$. This order reflects data dependencies: for every action $n$, the data that action $n$ relies upon must be produced by actions $m$ such that $m \prec n$. Actions that are not related by $\prec$ are called \emph{independent}.

A transaction can be partitioned (through a topological sort) into a sequence of \emph{layers} such that each layer consists of a set of independent actions, and for actions $m,n$ such that $m \prec n$, the layer containing $m$ occurs prior to the layer containing $n$ in the sequence. Formally, we have a set of \emph{layers}, $\{L_i\}$, to represent a cross-chain transaction and
\[\forall i, j, \ m \in L_i, \ n \in L_j : (m \prec n \Rightarrow i < j
)\]
where $L_i$ and $L_j$ are the $i$'th and $j$'th layer of the transaction, respectively while $m$ and $n$ are indexed actions.

% [k] Huaixi: making this a definition visually looks odd. Perhaps we can figure something out. Commenting for now to send to lawyers
The ideal (sequential and interference-free) execution proceeds as follows, shown in Figure~\ref{fig:tx_spec}. It begins by forming a checkpoint consisting of the current state of all blockchains participating in the transaction, represented as $cp$. At round $k$, operations on the nodes in layer $k$ ($L_k$) are issued in some order. (The actual order does not matter as the operations are independent.) If all operations in layer $k$  succeed, the results are saved and the next round (if any) is started. If some operation of layer $k$ fails, the execution is canceled and every blockchain is restored to its checkpointed state $cp$.

\begin{figure}
    \centering
    \includegraphics[width=\linewidth]{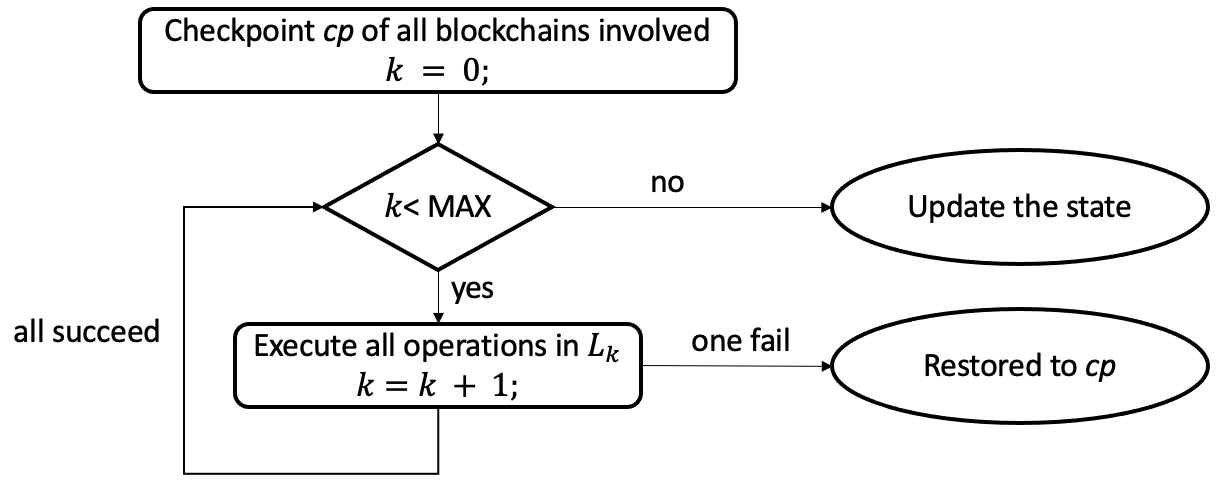}
      \caption{\label{fig:tx_spec} \textbf{The ideal execution for a cross-chain transaction}}
\vspace{-15pt}
\end{figure}

\subsection{Atomicity}

Multiple transactions may execute concurrently and may be interleaved with other blockchain actions issued independently of the transactions by other parties. The goal of the atomicity protocol is to ensure that in this concurrent, interference-prone setting, each transaction follows its ideal semantics as shown in Figure~\ref{fig:tx_spec}.

This is a two-phase protocol. A transaction is coordinated by a Proposer which is a smart contract on one of the blockchains. At a high level, the protocol follows the well known two-phase commit template for distributed database transactions. In the first phase, the proposer obtains agreement from all blockchains to “freeze” the portion of their state that is relevant to the operations in the transaction. In the second phase, the proposer executes operations from the transaction layer-by-layer, canceling the transaction if any individual operation fails.

We prove that protocol executions are strictly serializable. In terms of the closely related concept of linearizability, this guarantees that every transaction appears to take effect instantaneously, which provides atomicity with respect to other cross-chain transactions that follow this protocol and any contract operations that may be invoked independently. 

Two features sharply distinguish this protocol from two-phase commit protocols on databases. First, as the proposer and all operations operate on blockchains, we may assume that the likelihood of a machine failure is negligible, because of the protections provided by the replication that is built into blockchains. This is in sharp distinction to two-phase commit protocols on distributed databases, which are complex largely because they must handle various machine failure scenarios involving the Proposer and the databases. As a result, our atomicity protocol is significantly simpler than database commit protocols. 
The second distinction arises from the inherently trustless setting. Thus, any atomic protocol must contend with the potential for malicious behavior. This aspect is absent from traditional database protocols. It has two consequences. First, if the Proposer issues an operation on a remote blockchain, there is no guarantee that this operation will actually be executed. Our protocol therefore relies crucially on cross-chain bridges with the Secure Transfer property and on the abstract interface defined in the prior section, which ensures that operations will eventually be performed and cannot be corrupted. Second, our protocol relies on locking smart contract state. As the locking operations must be public, it is necessary to build in access controls to ensure that a malicious attacker cannot lock smart contracts to create a denial-of-service situation. 

To summarize, in the database domain, two-phase commit protocols focus on two types of failure: machine failures (i.e., machine crashes) and semantic failures (i.e., a failure within an operation, such as a divide-by-zero error). In our blockchain domain, machine failures may be ignored but, in their place, we must consider failures of trust, which arises from the potential for malicious attacks that aim to disrupt correctness or deny service.

\subsection{Blockchain Assumptions}
%%\commentaj{The paper talks about assumptions right above this. Then this heading "Assumptions" is confusing. Why not club all the assumtions at one place? This section is also talking about the chain properties only.}

We assume that two functionalities are provided by each blockchain that participates in transactions.

First, it must be possible to atomically checkpoint and lock a portion of the state of a single blockchain. We assume an abstract function $\lock(W,p)$ which atomically (1) saves the current values of the state variables $W$ to a copy $\checkpoint{W}$ of those variables and (2) blocks access to any action that modifies that state except when the action originates from the Proposer, denoted $p$.

%%\commentaj{Discuss this.}
Second, it must be possible to atomically restore and unlock the locked state.  We assume an abstract function $\unlock(W,f)$ which, when invoked by the Proposer $p$, atomically does the following: if $f$ is true (denoting a failure), the operation restores the state saved in $\checkpoint{W}$ to $W$; and, regardless of the value of $f$, the operation then removes the block on accessing the state in $W$.

In practice, these abstract operations may be implemented as follows. The variables $W$ represent the state variables of multiple smart contracts on the blockchain. To ensure atomicity for the lock and unlock operations and to prevent the misuse of the locking functionality, we define a \emph{transaction executor} contract 
%(shown in Figure~\ref{fig:bccp_bridge}) 
on each blockchain. The executor receives a transaction specification from the proposer and executes the transaction on behalf of the proposer. It also acts as a conduit for remote operation execution, as explained below. Each contract implements its own $\lock$ and $\unlock$ method. In addition, the code of the contract must ensure that every regular contract method is guarded to ensure that when the state variables of the contract are locked, the method can be invoked only by the transaction executor. Restricting the transaction execution to a known executor contract (which may be formally verified to ensure correct behavior) avoids the possibility of malicious behavior that can lock the state of participating contracts, creating a denial of service attack.

Pseudo-code illustrating the assumed smart contract structure is shown in Figure~\ref{fig:per-contract}. The syntax has its standard interpretation. The \verb|caller| variable represents the address of the party invoking the method; the \textbf{require} operation cancels execution if the value of its Boolean argument is false in the current contract state. For simplicity, the pseudo-code assumes that the methods in this contract do not invoke methods in other contracts. (Such indirect invocations are handled by relaying the address of the original external caller to the indirectly invoked methods.) 

\lstset{
  basicstyle=\small\ttfamily,
  keywordstyle=\color{black}\bfseries,%         \underbar,
  morekeywords={method,address,state,list,contract,boolean,require,in,and,or},
  morecomment=[l]{//}
}

\begin{figure}
  \centering
  \begin{lstlisting}
    contract C {
      s: state;                     
      checkpoint: state;
      owner: address;
      locked: boolean;              
      lockedBy: address;            
      trustedExecutors: list of address;   

      method addExecutor(address e) {
         require(caller = owner); 
         add e to the trustedExecutors list
      }
    
      method lock() {
        require(not(locked));
        require(caller in trustedExecutors);
        locked := True;
        lockedBy := caller;
        checkpoint := s;   // save 
      }

      method unlock(failure: boolean) {
        require(locked and caller = lockedBy);
        if (failure) {
          s := checkpoint;  // restore 
        } 
        locked := False;  
      }

      method m(...) { // guarded regular method
        require(not(locked) or caller = lockedBy);
        update state s 
      }
    }
  \end{lstlisting}
  \caption{\bf Sketch (pseudo-code) for the lock, unlock, and guarded regular methods of a contract.}
  \label{fig:per-contract}
\end{figure}

\subsection{The Atomicity Protocol}
At a high level, our protocol follows the template of a two-phase commit protocol. Such protocols were originally developed for distributed database transactions which, as explained, operate under a different model of trust and failure~\cite{bernstein2009principles, weikum2001transactional}.

\paragraph*{Transaction Execution} 
%%\hnote{is this duplicate with the transition T defined in blockchain?} [k] moved out of enumeration 
The structure of a cross-chain transaction $T$ is defined within an originating smart contract on some chain. Then, the transaction is sent to a the executor 
%\commentaj{Is this the executor of the proposed Atomicity Protocol?} 
contract on that chain. We refer to this executor as the Proposer and denote it by $p$. The Proposer communicates with executors on other chains using cross-chain bridges and the remote-call mechanism to freeze the state that is relevant to $T$, to execute individual operations, and to handle failures. Once the transaction is complete, the proposer notifies the originating smart contract.

\begin{figure}
    \centering
    \includegraphics[width=\linewidth]{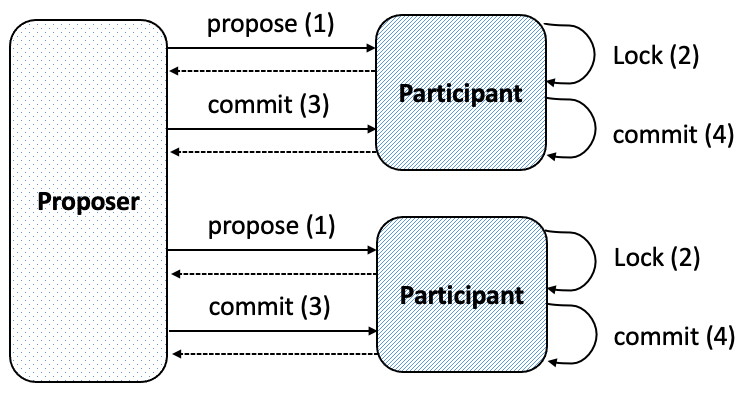}
      \caption{\label{fig:two_phase} \textbf{Overview of the two phase commit protocol.}}
\vspace{-15pt}
\end{figure}

\paragraph{Protocol Stages} The protocol consists of the following stages, while the first two stages are shown in Figure~\ref{fig:two_phase}.
\begin{enumerate}
\item{\bf First Phase}. The Proposer $p$ obtains agreement from all blockchains to “freeze” the state relevant to executing the transaction $T$ through step 1 and step 2 in the figure. 

Consider a blockchain $B$. Let $W$ be the union of $\scope(a)$ for every indexed action $(B,a)$ included in $T$. Through the remote-call mechanism, the proposer requests the trusted executor contract on blockchain $B$ (denoted $\texec(B)$) to checkpoint and lock the variables $W$ through the abstract $\lock(W,\texec(B))$ operation. In practice, $\scope(a)$ is a list of contracts on blockchain $B$. We assume that this list correctly defines the actual scope of action $a$, i.e., the contracts whose state may be modified either directly by the method $a$ or indirectly by contract methods invoked from $a$. (This transitive reach can be determined through automated code analysis.) The executor contract carries out the $\lock$ operation by invoking the \verb|lock| method in sequence on every contract in the supplied list. By the semantics of a blockchain, this sequence of method calls is executed atomically on that chain. 

  Only if the $\lock$ operation succeeds on every blockchain that is relevant to the transaction does the process move to the next phase. Otherwise, the transaction is canceled and the process moves to the failure mode defined below. The transaction may (potentially) be re-tried at a later point by the smart contract that originated the transaction.

\item{\bf Second Phase}. The proposer executes operations from $T$ in a series of rounds (repeating step 3 and step 4 in the figure). In the $k$'th round, the proposer invokes the operations defined in the $k$'th layer of the transaction $T$. The invocation of the indexed operation $(B,a)$ is carried out indirectly via the trusted executor $\texec(B)$ of blockchain $B$.

  In practice, this may be done by using the remote call mechanism to convey the action $a$ to $\texec(B)$, which then invokes the action on the contract where it is defined.

  The proposer waits until all actions in the $k$'th layer are complete. If every action completes successfully, the proposer moves to the next round. If not, the proposer enters the failure handling mode, described below. If all rounds are completed successfully, the proposer enters the success mode, described below. 
\end{enumerate}
  
  %%\commentaj{Formatting issue here. The bullet header says there are two stages but then we have 4 steps here. I am making a new para and bullet list. See if that's better. We can updated the para heading and opening statement.}

\paragraph{Transaction Outcomes} 
\begin{enumerate}
\item{\bf Failure}. The entire transaction is canceled. The proposer invokes $\unlock(W,\True)$ on each blockchain $B$ relevant to the transaction on which a $\lock(W)$ operation was successfully performed. This restores the state of that blockchain to its checkpointed value and unblocks the affected smart contract methods. In practice, this abstract operation is carried out by requesting the executor $\texec(B)$ on blockchain $B$ to invoke the \verb|unlock| method on every relevant contract with the failure parameter set to true. 

\item{\bf Success}. The transaction is committed. That is done by the proposer invoking $\unlock(W,\False)$ on each blockchain $B$ relevant to the transaction. In practice, the abstract operation is carried out by requesting the executor contract $\texec(B)$ on blockchain $B$ to invoke the \verb|unlock| method on every relevant contract with the failure parameter set to false.   
\end{enumerate}

\subsection{Correctness}
%% [kedar: dont follow the role of the fair exchange theorem]

We show that transactions may be viewed as being atomic. This follows from showing that executions are linearizable, in that each transaction appears to take effect (matching its ideal semantics) at some point between its start and finish.

%% [k] previous proof mixes up atomicity and failure/success . New proof separates the two. Same text, just differently arranged.
\begin{theorem}
  Transaction executions are secure and atomic. 
\end{theorem}
\begin{proof}
  Consider the execution of a transaction $T$. We first show that the execution of a transaction (whether it is complete and successful, or partial and unsuccessful) can be made atomic. We then show that the state of the participating blockchains is updated correctly on each outcome. 

  All state relevant to the operations in $T$ is locked at the start of the execution and only unlocked at the end of execution. Hence, it is not possible for an operation outside $T$ to interfere with the state transitions in $T$. The secure transfer property of cross-chain bridges ensures that all remote calls issued by the Proposer are eventually executed correctly, and that every call issued by the transaction executor on a chain seemingly on behalf of the Proposer must have actually originated at the Proposer. Thus, every operation in the transaction is eventually executed, and no operations other than those in the transaction are executed on the locked state. The transaction execution proceeds in the same sequence of rounds as the ideal execution. In each round, operations are independent; thus, concurrent execution results in the same state as any sequential order. 

  At each end of the transaction, the locking and unlocking processes proceed sequentially. Suppose that blockchains $B_0$ and $B_1$ are locked in that sequence. It is possible for other (non-$T$) operations to be performed on the state of $B_1$ after $B_0$ is locked but before $B_1$ is locked. Similarly, it is possible for other non-$T$ operations to be performed on the state of $B_1$ after $B_1$ is unlocked but before $B_0$ is unlocked (assuming that locks are removed in the reverse order). The first set of operations can be commuted to the left over the lock for $B_0$; while the second set of operations can be commuted to the right over the unlock of $B_0$. This defines an equivalent execution where the lock (similarly, unlock) operations occur in a contiguous block. 
  
  Any non-$T$ operations that are interleaved with the transaction operations must (by the definition of scope) operate on state that is disjoint from the state relevant to $T$. Hence, those operations can be commuted either to the left over the lock operations, or to the right over the unlock operations. The final result of the commutations is an equivalent computation where the operations of $T$ (including lock and unlock) are executed in a single atomic block.  
  
  Suppose that the execution of all operations in $T$ completes successfully. The success mode of the protocol ensures that the unlock operations ake the state changes permanent, as required by the atomicity semantics. If, however, execution of $T$ is partial, this must be due to a failed operation in some layer. The failure mode of the protocol ensures that the unlock operations restore the checkpointed state in each blockchain. Thus, in the failure case, the state in the scope of $T$ is unchanged, also as required by the atomicity semantics. 

  (Note: as all operations on a blockchain are visible, in the failure mode, the blockchain will record intermediate state changes as well as the operation that cancels those changes by restoring the checkpointed state. Hence, the failure of $T$ is visible in terms of the recorded blockchain operations; however, it is \emph{not visible} in terms of the semantics.)
  \end{proof}

There is, however, no guarantee that every issued transaction will eventually successfully acquire all relevant locks. One can easily construct pathological schedules where two transactions that must acquire locks on the same (or overlapping) sets of variables are issued together so that one transaction always loses the race. However, it is important to note that since every operation on a blockchain incurs a cost, such scenarios cannot persist indefinitely.

%%\knote{TODOs:
%%1. Paper format for CCS. Bring back page numbers.
%%2. Discussion of security. What is the attack model? How do the assumptions on bridges and our design decisions (e.g., adapters and executors) ensure security?
%%3. Measure the overhead (coding effort) for defining atomic transactions (in 10s of lines of code). 
%%4. Also revise Table 1 heading and labels for each row to make clear that those numbers represent the effort for implementing the abstractions. I.e., not the effort for implementing the transactions. 
%%4. Point out low effort needed in introduction. I.e., now that the abstraction layer is implemented, it requires only a handful of lines of code to create new transactions. 
%%5. Typo Figure 7.}

%\input{sec5_correctness}
\section{Implementation and Experiments}
%%\knote{Working through this section}
%In this section we need to define major components needed for real world execution of this system.

%A couple of lines summarizing those components

%You can also provide what (if) APIs these components are exposing or are expected to expose or should expose

%Any major back compatibility issues must be describe here too.

%If your design is forward compatible then shout loudly about it.

To demonstrate practicality of the proposed 
%%\ajcomment{it will be better to identify the model clearly like "atomic protocol".} 
model and atomic protocol, we wrote a suite of smart contracts in Solidity. We further deployed them across four different testnets\footnote{Testnets are functionally similar to production blockchains the only difference is that the tokens on testnets do not have financial transaction value. Testnets are used for testing and research purposes.} including Mumbai, Fantom, Fuji, and Goerli. However, due to the limitation of test token amount, we only tested the applications on two testnets listed in Table.~\ref{tb:exp}. We used Hardhat~\cite{hardhat} to deploy smart contracts on the testnets.

% \editaj
% {Moreover, we brought to life two multi-chain applications: an atomic swap between two chains and an exchange process encompassing three chains.
% We tested the delay of one transaction and the amount of gas consumption in total for these two applications under different scenarios.
% All experimental details will be shown in subsection.xx.}
Further, we developed following two multi-chain test applications: (1) An atomic swap between two chains and (2) An exchange process encompassing three chains.
We studied the transaction latency and gas consumption for the test applications to show the applicability of our work. Details of the results are in \S\ref{sec:impl:results}. 

%%[k] omitted as we may need permission to make the code accessible, even for review.
%%\footnote{\hnote{Source code for all the designs and experiments is available at https://anonymous.4open.science/r/SCCP-9CE8.}}.

%\subsection{Solidity Smart Contracts}
\subsection{Implementation}
Bridges such as LayerZero and IBC protocol typically host endpoint contracts within each blockchain. These endpoints facilitate communication with other blockchains. A challenge arises due to the heterogeneity of these endpoint interfaces across different bridges. To navigate this, we introduce a unified converter interface applicable to all bridges. The structure of our implementation (in Solidity) is shown in Figure~\ref{fig:bccp_bridge}. The interface includes functions such as \emph{notify}, \emph{notify\_ack}, and \emph{remote\_call}, which are  realized using the lower-level bridge functions. 
%%In addition, we also define a name variable in the converter for locating the bridge in use. 
Once a bridge converter is defined   on a blockchain, the abstraction simply requires the inclusion of that converter's address within its contract. This integration allows the abstraction to employ a standardized communication interface, enabling the use of various bridges by invoking distinct bridge converters as needed. Though one interface function may involve several transactions in different chains, the complicated payment for multiple communication is abstracted away, making it possible for an application designer to focus on the application logic.

\begin{figure}
    \centering
    \includegraphics[width = \linewidth]{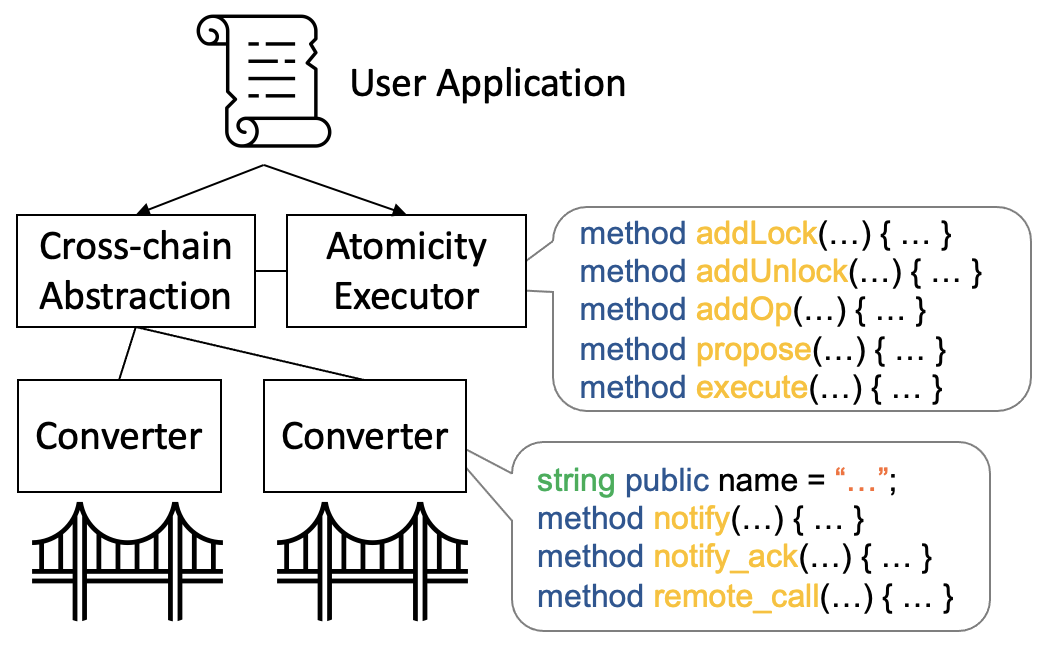}
      \caption{\label{fig:bccp_bridge} \textbf{Implementation Framework.}}

\end{figure}

The implementation of atomic transactions builds on this communication layer. The structure is defined by public function calls such as \emph{addLock}, \emph{addUnlock}, \emph{addOp}, \emph{propose}, and \emph{execute}. %%Application developers can effortlessly define all the necessary lock/abort/execute/unlock operations via these functions. 
Application developers must deposit a specified amount to support payment for bridge communications that occur during the transaction execution. 
% into the bridge for any subsequent transactions followed by the initial transaction (making sure that all communication among chains can be done). By using our atomic contract, it's ensured that the transaction is either executed across all blockchains or universally aborted. 
Notably, within the atomic execution smart contract, storage provisions are made to retain details of these multi-chain applications. Once an application is documented, it's stored permanently, allowing for subsequent initialization without the redundancy of re-defining all requisite operations.

The sizes of the contracts that make up the implementation of our abstractions, and the deployment cost on Testnet Fantom, are shown in Table.~\ref{tb:solidity}. The abstraction and atomicity smart contracts only need to be deployed once on the blockchain. 
%They are programmed so that  manage different messages from different bridges. 
The implementation is flexible and supports multiple bridges. Supporting a new bridge for connection with another chain  requires the creation of a converter (approximately 242 lines of code), and a function call to register the bridge with the converter. 
%and a single function call in the abstraction to incorporate this converter.

\begin{table}[]
\caption{Statistics for Deploying the Atomicity Protocol Smart Contracts in Fantom}
\label{tb:solidity}
\resizebox{\linewidth}{!}{%
\begin{tabular}{|c|c|c|}
\hline
\textbf{Smart Contract} & \textbf{Size (LoC)} & \textbf{Deploy gas} \\ \hline
\begin{tabular}[c]{@{}c@{}}Convertor\\ (LayerZero)\end{tabular} & 242 & 1.4 M \\ \hline
Abstract Bridge & 601 & 2.7 M \\ \hline
Atomicity Executor & 225 & 1.6 M \\ \hline
\end{tabular}%
}
\end{table}

\begin{table}[]
\caption{Overhead (coding effort) for defining atomic transactions}
\label{tb:solidity_app}
\resizebox{\linewidth}{!}{%
\begin{tabular}{|c|c|c|}
\hline
\textbf{Contract} & Atomic Swap & Three-Chain Exchange \\ \hline
\textbf{Overhead (LoC)} & 10 & 14 \\ \hline
\end{tabular}%
}
\end{table}

%Once the essential infrastructure for the atomicity protocol is deployed, only a 

From a programmer's viewpoint, this  abstraction layer considerably simplifies the process of constructing atomic cross-chain transactions. Only a few lines of code are needed to implement an atomic swap and a three-chain exchange, as shown in Table.~\ref{tb:solidity_app}. This code sets up the structure of the transaction (its operations and their ordering), using the interface shown in  Figure~\ref{fig:bccp_bridge}. 

% shows the overhead of coding effort needed to make the test application atomic over multiple chains, including a two-chain swap and an exchange among three distinct chains. Both require less than 20 LoC to define the atomic cross-chain transaction. This additional effort predominantly involves method calls within the Atomicity Executor, as depicted in 
% To execute a single cross-chain transaction atomically, manual inclusion of all requisite operations is necessary. These operations include  \emph{lock}, \emph{unlock}, \emph{propose} and \emph{execute}, ensuring the transaction's execution is correctly managed by our atomicity protocol.}

\subsection{Experiments}
We first executed a two-chain atomic swap between the Mumbai and Fantom testnets. Mumbai serves as Polygon's testnet, mirroring the mainnet functionalities of Polygon. In this process, Fantom testnet plays the role of the proposer. For a successful transaction on the Mumbai blockchain, three operations are essential: lock, transfer, and unlock. Each operation requires a remote call from the fantom testnet and sends back an acknowledgement to indicate the status of execution, which is implemented through the \emph{remote\_call} in the abstraction. 
The atomic swap transaction, when successful, undergoes the following steps:
\begin{itemize}
    \item Step 1: Fantom proposes the transaction to Mumbai and self-locks.
    \item Step 2: Mumbai consents, then locks, and dispatches an acknowledgment.
    \item Step 3: Fantom initiates a locked-mode transfer while prompting Mumbai to follow suit.
    \item Step 4: Mumbai completes the transfer and sends an acknowledgment.
    \item Step 5: Fantom initiates the unlock process, directing Mumbai to do the same.
    \item Step 6: Mumbai concludes the unlock phase and communicates acknowledgment.
\end{itemize}
As a result, a successful atomic transfer requires 6 messages between two blockchains.

In addition to the successful case, we experimented with two potential failure scenarios to gauge the resilience of our protocol. The first scenario is a lock conflict in Step 2, where another transaction already claims the lock. The second involves an unsuccessful transfer due to inadequate funds during Step 4. Both situations prompt a transaction termination (Abort), reverting both blockchains to their pre-transaction states. The outcomes of these tests were in line with our expectations.

\begin{figure}
    \centering
    \includegraphics[width = 0.9\linewidth]{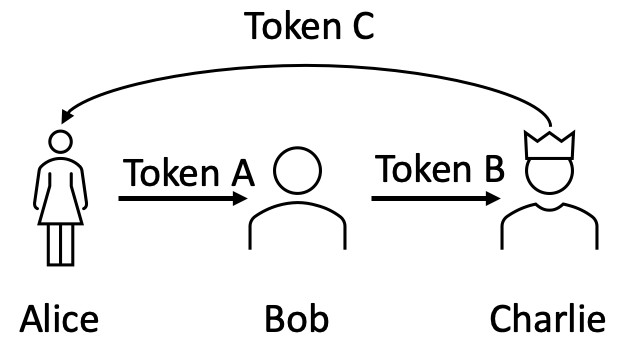}
      \caption{\label{fig:multi_chain} \textbf{Three-Chain Exchange.}}

\end{figure}

We also explored a multi-chain scenario involving exchanges across three distinct chains. Imagine Alice possesses token A and aims to trade it for token C with Charlie. However, Charlie lacks an account on Chain A and only accepts token B as a valid trade for token C. Therefore, Alice transfers token A to Bob, who then gives token B to Charlie, resulting in Charlie providing token C to Alice. This scenario is visually depicted in Figure.~\ref{fig:multi_chain}. Since three different tokens are used, three different chains are involved. Such multi-chain exchanges frequently occur in supply chain management systems, as noted in ~\cite{augusto2023multi}, emphasizing their growing significance in the evolving blockchain landscape.

Though we have successfully deployed contracts across multiple chains and conducted localized function tests, certain multi-chain applications remain challenging due to a shortage of testnet tokens. For instance, following certain updates that ascribed real value to the Goerli token, contract deployment on Goerli became considerably costly. Consequently, for our three-chain exchange experiment, we simulated the application on a single chain, imitating two distinct chains. As cross-chain communication is exclusive between the proposer and participants and absent among participant chains, this approximation does not compromise protocol integrity.

\subsection{Results} \label{sec:impl:results}
\begin{table*}[]
\caption{Measurements for Various Cross-chain Transactions. (Column headings: XC=cross-chain; \# of tx is the number of smart contract transactions; Total Fee/Token refers to the fee in terms of the cryptocurrency tokens for that chain; Bridge Gas is the gas cost for  bridge operations; Atomic Gas is the gas cost for  atomic operations.)}
%Atomic Application on Real Blockchain Statistics}
\label{tb:exp}
\resizebox{\textwidth}{!}{%
\begin{tabular}{|c|cl|ccccccc|}
\hline
\multirow{2}{*}{\textbf{\begin{tabular}[c]{@{}c@{}}Multi-Chain\\ Apps\end{tabular}}} & \multicolumn{2}{c|}{\multirow{2}{*}{\textbf{\begin{tabular}[c]{@{}c@{}}Test Chain\\ (L: Proposer\\ P: Participants)\end{tabular}}}} & \multicolumn{7}{c|}{\textbf{Transaction Statistics}} \\ \cline{4-10} 
 & \multicolumn{2}{c|}{} & \begin{tabular}[c]{@{}c@{}}\# of XC \\ msgs\end{tabular} & \begin{tabular}[c]{@{}c@{}}\# of \\ tx\end{tabular} & \begin{tabular}[c]{@{}c@{}}Total\\ Gas\end{tabular} & \begin{tabular}[c]{@{}c@{}}Total Fee\\ / Token\end{tabular} & \begin{tabular}[c]{@{}c@{}}Bridge\\ Gas\end{tabular} & \begin{tabular}[c]{@{}c@{}}Atomic\\ Gas\end{tabular} & \begin{tabular}[c]{@{}c@{}}Avg.\\ Time\end{tabular} \\ \hline
\multirow{2}{*}{\textbf{Atomic Swap}} & \multicolumn{2}{c|}{L: Fantom} & 3 & 4 & 2.14M & 0.003 & 1.11M & \multicolumn{1}{c|}{1.03M} & \multirow{2}{*}{301s} \\ \cline{2-3}
 & \multicolumn{2}{c|}{P: Mumbai} & 3 & 3 & 1.38M & 0.003 & 1.02M & \multicolumn{1}{c|}{0.36M} &  \\ \hline
\multirow{2}{*}{\textbf{\begin{tabular}[c]{@{}c@{}}Atomic Swap \\ (Lock Fail)\end{tabular}}} & \multicolumn{2}{c|}{L: Fantom} & 2 & 3 & 1.48M & 0.001 & 0.74M & \multicolumn{1}{c|}{0.74M} & \multirow{2}{*}{196s} \\ \cline{2-3}
 & \multicolumn{2}{c|}{P: Mumbai} & 2 & 2 & 0.85M & 0.002 & 0.68M & \multicolumn{1}{c|}{0.17M} &  \\ \hline
\multirow{2}{*}{\textbf{\begin{tabular}[c]{@{}c@{}}Atomic Swap\\ (Update Fail)\end{tabular}}} & \multicolumn{2}{c|}{L: Fantom} & 3 & 4 & 2.13M & 0.003 & 1.11M & \multicolumn{1}{c|}{1.02M} & \multirow{2}{*}{283s} \\ \cline{2-3}
 & \multicolumn{2}{c|}{P Mumbai} & 3 & 3 & 1.37M & 0.003 & 1.02M & \multicolumn{1}{c|}{0.35M} &  \\ \hline
\multirow{3}{*}{\textbf{\begin{tabular}[c]{@{}c@{}}Three \\ Exchange\end{tabular}}} & \multicolumn{2}{c|}{L: Fantom} & 6 & 4 & 4.33M & 0.006 & 2.22M & \multicolumn{1}{c|}{2.11M} & \multirow{3}{*}{318s} \\ \cline{2-3}
 & \multicolumn{2}{c|}{P: Mumbai-1} & 3 & 3 & 1.53M & 0.003 & 1.02M & \multicolumn{1}{c|}{0.51M} &  \\ \cline{2-3}
 & \multicolumn{2}{c|}{P: Mumbai-2} & 3 & 3 & 1.53M & 0.003 & 1.02M & \multicolumn{1}{c|}{0.51M} &  \\ \hline
\end{tabular}%
}
\end{table*}

The results of our experiments can be found in Table~\ref{tb:exp}. Fantom acts as the proposer, while Mumbai serves as the participants. For each transaction, we documented the number of cross-chain messages invoked by the LayerZero bridge, as well as the number of operations performed on both chains. 

Gas is essential for the execution of all these operations. When calculating gas consumption, two major components come into play:
1) Gas required for bridge functionalities, which include sending and validating messages.
2) Gas necessary for atomic operations.
To delineate these different gas requirements, we designed a rudimentary cross-chain application that uses a recursive call with an empty payload. Through this, we measured the basic communication gas consumption between Fantom and Mumbai. Any additional gas was attributed to atomic operations.

Analyzing the data from Table.~\ref{tb:exp}, we observe that the gas consumption for atomic operations is generally less than that for bridge communications. In all listed scenarios, atomic operations account for as little as 20 percent to a maximum of 50 percent of the total gas. This indicates that the primary gas consumption in atomic transactions predominantly arises from bridge communication, while the overhead from atomic operations remains small.

Furthermore, the main contributors to transaction latency are the bridge communication between the two chains and the time required for blockchain confirmations. This is largely due to the atomicity protocol being integrated within on-chain smart contracts. To acquire a more accurate understanding of the average transaction time, we executed a single transaction multiple times and computed the mean duration for a complete transaction. According to the results in Table.~\ref{tb:exp}, transactions involving three communication rounds (as seen in cases 1, 2, and 4) typically conclude in about 5 minutes. In contrast, scenarios like the lock failure in atomic swap, which involve only two communication rounds, might complete in just 3 minutes. These timeframes, spanning a few minutes, are reasonable when compared with other cross-chain transaction durations.

\subsubsection{Atomic Swaps}
Atomic swaps have become a staple in modern cross-chain communications. This mechanism ensures seamless and trustless exchanges between two blockchains without the need for intermediaries. The essence of atomicity in these swaps can often be embedded directly within the design of the bridge facilitating the exchange. Typically, a successful swap necessitates only two communication rounds. However, our approach diverges slightly by introducing an additional communication round. While this might seem more cumbersome compared to the prevalent methods, our methodology offers broader compatibility and adaptability to decentralized applications, ensuring they operate with enhanced reliability and security.

\subsubsection{Three-Way Exchange}
The current landscape of blockchain technology lacks support for multi-chain applications. This evident gap serves as a primary motivation behind this work. By integrating an atomicity layer atop existing bridges, our mechanism considerably simplifies the development of complex multi-chain applications. A compelling example is the three-way exchange, which is challenging to program from scratch, but is realized with just $14$ lines of code using our framework. 

% multi-chain operations but also guarantees atomicity for such intricate applications. This advancement ensures that transactions across multiple chains occur seamlessly and are either fully executed or entirely aborted, thus preserving the integrity of the operations. While there's an overhead introduced by this additional layer, it remains within a reasonable range, ensuring that the efficiency of the bridge isn't significantly compromised. This innovation paves the way for more complex and interconnected blockchain ecosystems in the future.

%\commentaj{Do you have some qualitative/quantitative results to show that the contract written in Solidity using the proposed design is similar or not too bad in terms of complexity than the usual program?}
\section{Related Work}
As blockchain interactions across different platforms become increasingly prevalent, the challenge of varying security assumptions among these blockchains gains a lot of attention. To address this, weaker blockchain may use checkpointing solutions, as discussed in~\cite{karakostas2021securing, sankagiri2021blockchain}, which boost the trust using a stronger blockchain. However, these solutions require the weaker blockchain to give up sovereignty. Alternatively, as shown in ~\cite{wang2022trustboost}, protocols built upon the cross-chain bridges can be deployed to create a combined ledger with boosted trust for weaker blockchains. Our work uses similar idea with~\cite{wang2022trustboost} to have protocols atop bridges but to provide atomicity for multi-chain transactions.

Atomicity in transactions spanning multiple blockchains has received great attention. For example, the work in~\cite{xue2023fault, herlihy2019cross} formulated and solved the application of asset exchange involving more than two chains. There are two major differences in our work. 
First, our method has a substantially broader scope of applicability. We formulate and solve the question of atomically performing an \emph{arbitrary} sequence of \emph{arbitrary} smart contract transactions. Asset exchange is one specific instance of this general framework.
Second, the underlying solution approach is quite different. The work in~\cite{herlihy2019cross} presents two solutions: the first is based on a construct called a ``timed lock," the second on a custom protocol (CBC). While our protocol is similar in spirit to their CBC protocol, as both are based on the general two-phase template used in distributed databases, there are differences in scope and assumptions. Their protocol is designed specifically for the asset transfer problem, e.g., the ``escrow'' step in CBC escrows an asset, while the locking mechanism in our method blocks operations on an arbitrarily defined state. The paper discusses the issue of attackers forging transactions but does not fully define a protocol to prove that claimed transactions have indeed occurred on the CBC chain. Our solution relies on the general notion of a cross-chain bridge, which provides precisely this guarantee in the form of the secure transfer property. Partitioning the solution in this manner  allows for the atomicity protocols to benefit from improvements in bridge implementation, and enhances portability across bridges. 

Our protocols rely crucially on bridges with the secure transfer property. We use the well known zkbridge~\cite{xie2022zkbridge} and LayerZero~\cite{zarick2021layerzero}) in our prototype implementation. 

%\hnote{TODO: add some related work in database or distributed system for two phase commit protocol?}
\section{Conclusions}
Blockchain interoperability is already a significant problem and one that we expect will become even more relevant as data is spread across multiple blockchains. Thus, it is important to build high-level abstractions that simplify programming and increase productivity. 

Towards this goal, our work introduces a standardized high-level interface for efficient blockchain communication, enhancing cross-chain portability. Our work also defines a broadly applicable, multi-blockchain transaction protocol, which ensures atomicity for applications spanning multiple chains. For both protocols, we rigorously formalize and prove the correctness and security of these protocols under clear assumptions on blockchain and bridge behavior and smart contract structure. 

We show that these protocols are easily implementable. Our prototype implementation is programmed in Solidity, utilizing LayerZero (mostly) and IBC bridges. This prototype, consisting of both the high-level interface and the atomicity protocol, has been rigorously tested on testnets. The successful application of our atomicity protocol in executing both pairwise and multiway cryptocurrency swaps on these testnets demonstrates its practical effectiveness and broad applicability.

Overall, our work opens up new possibilities for general multi-chain transactions and, we believe, sets the stage for further research in this area. An important direction for further research is to discover finer-grained structure in multi-chain transactions that can be exploited to optimize storage and communication costs. 

%%such as better  of the multi-chain transaction to reduce the storage cost. The practical uses we hve demonstrated show the great promise of our protocols, leading the way for more creative and effective solutions in blockchain technology.

%\newpage

\bibliographystyle{ACM-Reference-Format}
\bibliography{
references
}

\end{document}